\documentclass[journal]{IEEEtran}
\IEEEoverridecommandlockouts

\ifCLASSINFOpdf
   \usepackage[pdftex]{graphicx}
  \graphicspath{{figures/}}
  \usepackage{epstopdf}
  \DeclareGraphicsExtensions{.pdf,.jpeg,.png}
\else
   \usepackage[dvips]{graphicx}
   \graphicspath{{figures/}}
   \DeclareGraphicsExtensions{.eps}
\fi

\usepackage{pgfplots}
  \pgfplotsset{compat=newest}
  \usetikzlibrary{plotmarks}
  \usetikzlibrary{arrows.meta}
  \usepgfplotslibrary{patchplots}
  \usepackage{grffile}
  \pgfplotsset{plot coordinates/math parser=false}
  \newlength\figureheight
  \newlength\figurewidth

\usepackage{pdfpages}

\usepackage{stfloats}
\usepackage{nicefrac}
\usepackage{amsthm}
\usepackage{amssymb}
\usepackage{pifont}

\usepackage{latexsym}
\usepackage{times}
\usepackage{cite}
\usepackage{array}
\usepackage[cmex10]{amsmath}
\usepackage{multicol}
\usepackage{urwchancal}
\usepackage[tight,footnotesize]{subfigure}
\usepackage{multirow}
\usepackage{xcolor}

\usepackage[nolist,printonlyused]{acronym}      

\usepackage{optidef} 
\usepackage{algorithm,algorithmic}

\newtheorem{theorem}{Theorem} 

\newtheorem{observation}{Observation}

\usepackage{xpatch}
\makeatletter
\xpatchcmd{\proof}{\@addpunct{.}}{\normalfont\,\@addpunct{:}}{}{}
\makeatother

\newcommand\copyrighttext{%
  \footnotesize \textcopyright 2019 IEEE. Personal use of this material is permitted. Permission from IEEE must be obtained for all other uses, in any current or future media, including reprinting/republishing this material for advertising or promotional purposes, creating new collective works, for resale or redistribution to servers or lists, or reuse of any copyrighted component of this work in other works.
  DOI:} 
\newcommand\copyrightnotice{%
\begin{tikzpicture}[remember picture,overlay]
\node[anchor=south,yshift=8pt] at (current page.south) {\fbox{\parbox{\dimexpr\textwidth-\fboxsep-\fboxrule\relax}{\copyrighttext}}};
\end{tikzpicture}%
}

\DeclareMathAlphabet\mathbfcal{OMS}{cmsy}{b}{n}

\newcounter{subeqsave}
\newcommand{\savesubeqnumber}{\setcounter{subeqsave}{\value{equation}}%
\typeout{AAA\theequation.\theparentequation}}
\newcommand{\recallsubeqnumber}{%
  \setcounter{equation}{\value{subeqsave}}\stepcounter{equation}}

\begin{document}

\title{Low-Complexity OFDM Spectral Precoding}
\author{\IEEEauthorblockN{Shashi Kant$^{(\dagger,\ddagger)}$, Gabor Fodor$^{(\dagger,\ddagger)}$, Mats Bengtsson$^{(\ddagger)}$, Bo G\"oransson}$^{(\dagger)}$, and Carlo Fischione$^{(\ddagger)}$\\
\IEEEauthorblockA{$^{(\dagger)}$Ericsson AB, Kista 164 80, Sweden.}  \\
\IEEEauthorblockA{$^{(\ddagger)}$KTH Royal Institute of Technology, Stockholm 114 28, Sweden.}\\
\{shashi.v.kant, gabor.fodor, bo.goransson\}@ericsson.com, mats.bengtsson@ee.kth.se, carlofi@kth.se
\vspace{-5mm}
}

\maketitle
\copyrightnotice

\renewcommand\qedsymbol{$\blacksquare$}

\iftrue
\renewcommand{\vec}[1]{\ensuremath{\boldsymbol{#1}}}
\newcommand{\mat}[1]{\ensuremath{\boldsymbol{#1}}}
\newcommand{\herm}{{\rm H}}
\newcommand{\tran}{{\rm T}}
\newcommand{\trans}{{\rm T}}
\newcommand{\trace}{{\rm Tr}}
\newcommand{\diag}{{\rm diag}}
\newcommand{\rank}{{\rm rank}}
\newcommand{\EVM}{{\rm EVM}}
\newcommand{\SNR}{{\rm SNR}}
\newcommand{\SINR}{{\rm SINR}}
\newcommand{\expect}{\mathbb{E}}
\newcommand{\Cm}{\mathbb{C}}
\newcommand{\Rm}{\mathbb{R}}
\newcommand{\CN}{\mathcal{CN}}
\newcommand{\be}{\begin{equation}}
\newcommand{\ee}{\end{equation}}
\newcommand{\pdf}{\mathcal{P}}
\newcommand{\prox}{\rm{prox}}
\newcommand{\dom}{\rm{dom}}
\newcommand{\ie}{\textit{i.e.}}
\newcommand{\eg}{\textit{e.g.}}
\newcommand{\cf}{\textit{cf.}}

\newcommand{\NR}{\ensuremath{N_{\rm R}}}
\newcommand{\NT}{\ensuremath{N_{\rm T}}}
\newcommand{\NL}{\ensuremath{N_{\rm L}}}
\newcommand{\NCP}{\ensuremath{N_{\rm CP}}}

\newcommand{\A}{\ensuremath{\boldsymbol{A}}}
\renewcommand{\a}{\ensuremath{\boldsymbol{a}}}
\newcommand{\B}{\ensuremath{\boldsymbol{B}}}
\renewcommand{\b}{\ensuremath{\boldsymbol{b}}}
\newcommand{\C}{\ensuremath{\boldsymbol{C}}}
\renewcommand{\c}{\ensuremath{\boldsymbol{c}}}

\newcommand{\D}{\ensuremath{\boldsymbol{D}}}
\renewcommand{\d}{\ensuremath{\boldsymbol{d}}}
\newcommand{\E}{\ensuremath{\boldsymbol{E}}}
\newcommand{\F}{\ensuremath{\boldsymbol{F}}}

\newcommand{\K}{\ensuremath{\boldsymbol{K}}}

\newcommand{\M}{{\mathbf{M}}}
\newcommand{\I}{\ensuremath{\boldsymbol{I}}}

\newcommand{\R}{\ensuremath{\boldsymbol{R}}}

\renewcommand{\S}{\ensuremath{\boldsymbol{S}}}

\newcommand{\T}{\ensuremath{\boldsymbol{T}}}
\renewcommand{\t}{\ensuremath{\boldsymbol{t}}}

\newcommand{\U}{\ensuremath{\boldsymbol{U}}}
\renewcommand{\u}{\ensuremath{\boldsymbol{u}}}

\newcommand{\V}{\ensuremath{\boldsymbol{V}}}
\renewcommand{\v}{\ensuremath{\boldsymbol{v}}}

\newcommand{\X}{\ensuremath{\boldsymbol{X}}}

\newcommand{\Y}{\ensuremath{\boldsymbol{Y}}}

\newcommand{\0}{\ensuremath{\boldsymbol{0}}}

\DeclareRobustCommand{\BigOh}{%
  \text{\usefont{OMS}{cmsy}{m}{n}O}%
}

\fi

\newcommand{\sk}[1]{{\color{black} #1}} 
\newcommand{\comments}[1]{{\color{magenta} #1}}
\def\baselinestretch{1}
\def\gf#1{\textcolor{black}{#1}}
\newcommand{\mx}[1]{\mathbf{#1}}

\def\baselinestretch{1}

\begin{abstract} 
This paper proposes a new large-scale mask-compliant spectral precoder (LS-MSP) for orthogonal frequency division multiplexing systems.
In this paper, we first consider a previously proposed mask-compliant spectral precoding scheme that 
utilizes a generic convex optimization solver which suffers from high computational complexity, notably in large-scale systems. 
To mitigate the complexity of computing the LS-MSP, 
we propose a divide-and-conquer approach that breaks the original problem into 
smaller rank 1 quadratic-constraint problems and each small problem yields closed-form solution. 
Based on these solutions, we develop three specialized first-order low-complexity algorithms, based on 1) projection on convex sets and 2) the alternating direction method of multipliers. We also develop an algorithm that capitalizes on the closed-form solutions for the rank 1 quadratic constraints, which is referred to as 3) semi-analytical spectral precoding. Numerical results show that the proposed LS-MSP techniques outperform previously proposed techniques in terms of the computational burden while complying with the spectrum mask. The results also indicate that 3) typically needs 3 iterations to achieve similar results as 1) and 2) at the expense of a  slightly increased computational complexity.
\end{abstract}

\section{Introduction}
The requirement space of the fifth generation (5G) wireless networks new radio (NR) is driven by several use-cases, including  mobile broadband and internet-of-things applications. 
One of the challenges to cater to these services is handling the spectrum crunch. 
A vital technology enabler to deploy 5G services in the lower frequency bands
is dynamic spectrum sharing between 5G and other already deployed radio access technologies. 
An orthogonal frequency division multiplexing (OFDM) system offers flexibility and facilitates dynamic spectrum sharing. 
Indeed, 5G NR systems employ OFDM for both downlink and uplink transmissions like the downlink transmission of fourth-generation networks \cite{Dahlman5GNR2018}.

Although OFDM systems offer various favourable properties, 
such as robustness to multipath fading 
-- enabling low-complexity receivers and achieving high spectral efficiency -- 
it is well-known that OFDM waveform suffers from high out-of-band emissions (OOBE), \eg, see \cite{Dahlman5GNR2018}. 
The leading cause for this is the
discontinuities at the boundaries of the rectangular window of the generated OFDM symbols. 
Unfortunately, the high OOBE causes significant interference to the neighbouring channel occupants unless the OOBE is adequately suppressed. 
Therefore, NR systems are designed to comply with well-defined adjacent channel leakage ratio (ACLR) and spectral emission mask (SEM) requirements, see for example \cite{3GPPTS38.1042018NRReception}.

Spectral precoding is one of the promising techniques for OOBE reduction, \sk{by spectrally precoding the complex data symbols} prior to OFDM modulation without increasing the delay/time dispersion or penalizing the \sk{cyclic prefix} of the transmitted signal. 
In \cite{DeBeek2009SculptingPrecoder}, a notch-based spectral precoder (NSP) scheme is proposed by nulling the OOBE at predefined frequency points while minimizing the Euclidean distance between the precoded and the original data vector, \sk{ that is equivalent to an error-vector-magnitude (EVM). 
The NSP method renders high EVM, particularly at the edge subcarriers, leading to increased bit error rate (BER) / block error rate (BLER) at the receiver}. 
Recognizing these problems, in \cite{Kumar2015WeightedRadio} a weighted-NSP scheme is proposed to reduce the distortion power level at the edge subcarriers by spreading the distortion power over the allocated subcarriers. 
An interesting approach proposed in \cite{Tom2013MaskShaping} is the spectral precoding that only needs to fulfil the desired SEM, which is referred to as mask-compliant spectral precoder (MSP). 
Moreover, MSP offers improved performance in terms of BLER as compared with NSP depending on the target SEM at the expense of increased complexity. \\
The higher complexity of MSP is due to the convex optimization formulation that typically does not yield a closed-form solution. 
Therefore, the authors of \cite{Tom2013MaskShaping} suggest utilizing a generic optimization solver, which typically employs interior-point methods \cite{Boyd2004ConvexOptimization}. 
Instead of computing a spectrally precoded symbol vector in MSP \cite{Tom2013MaskShaping}, 
the authors in \cite{Kumar2016} proposed optimal and sub-optimal schemes to obtain a spectral precoder matrix that can fulfil the target SEM by multiplying the data symbol vector by a fixed precoding matrix, 
while offloading the major computational part offline. However, \gf{supporting large bandwidths along with a large number of antennas while enabling dynamic spectrum sharing, the memory required for the storage of large spectral precoder matrices is a bottleneck in real systems. Thus, the authors in \cite{Kumar2016} proposed a suboptimal approach. Therefore, in contrast to \cite{Kumar2016}, we not only seek low-complexity and memory-efficient but also (near) optimal precoding schemes for the large-scale MSP (LS-MSP) that fulfils the target mask, and minimizes the EVM.}\\
\sk{\textit{Contributions}: We decompose the LS-MSP optimization problem into subproblems, where each subproblem yields closed-form solution}. 
We also develop three specialized low-complexity algorithms that find 
the solution to the LS-MSP problem based on utilizing 
1) projection on convex sets (POCS) and 
2) consensus alternating direction method of multipliers (ADMM) techniques \cite{Combettes2011}\cite{Parikh2013}\cite{Boyd2010}. 
The third specialized algorithm, dubbed as 3) semi-analytical spectral precoding (SSP), is derived by utilizing the closed-form solution of a dual variable based on the coordinate descent scheme \cite{Luo1992}. \gf{We would also like to highlight, that the suboptimal approach in \cite{Kumar2016} has similar computational cost as in our SSP algorithm; but the  SSP converges in 2-3 iterations, faster than the suboptimal scheme in \cite{Kumar2016}, and yet SSP offers (near) optimal performance.}

\textit{Notation}:
The $i$-th element of a vector $\vec{a}$ is denoted by $a_i$, and element in the $i$-th row and $j$-th column of the matrix $\mat{A}$ is denoted by $\mat{A}\left[i,j\right]$. The $i$-th row vector of a matrix $\mat{A}$ is represented as $\mat{A}\left[i,:\right]$. The transpose and conjugate transpose of a vector or matrix are denoted by $\left(\cdot\right)^{\rm T}$ and $\left(\cdot\right)^{\herm}$, respectively. The complex conjugate to a vector $\vec{b}$ is denoted by $\vec{b}^*$. The set of complex and real numbers are denoted by $\mathbb{C}$ and $\mathbb{R}$, respectively; $\Re\{x\}$ shows the real part of a complex number $x$. An $i$-th iterative update is denoted by $(\cdot)^{(i)}$. 

\section{Preliminaries} \label{sec:prelim}
This section introduces the downlink system model followed by a brief description of the OOBE in OFDM and conventional spectral precoders.

\subsection{System Model and Out-of-Band Emissions} 
\label{sec:system_model}
We consider a single transmit and receive antenna based OFDM system as depicted in Fig. \ref{fig1__block_diag_tx_rx}. 
A given OFDM symbol vector in the frequency-domain comprises $N$ subcarriers denoted as 
$\vec{d}  = 
\left[x_0,\ldots, x_{N-1} \right]^{\tran} \! \in \! \mathbb{C}^{N \times 1}$, 
where $k$-th subcarrier is modulated by, 
\eg, quadrature amplitude modulation (QAM) alphabet, \sk{for instance,  see \cite{Dahlman5GNR2018}}.

\begin{figure*}[!htbp]
\begin{center}
\scalebox{0.8}{\includegraphics[trim=0 0 0 0, clip]{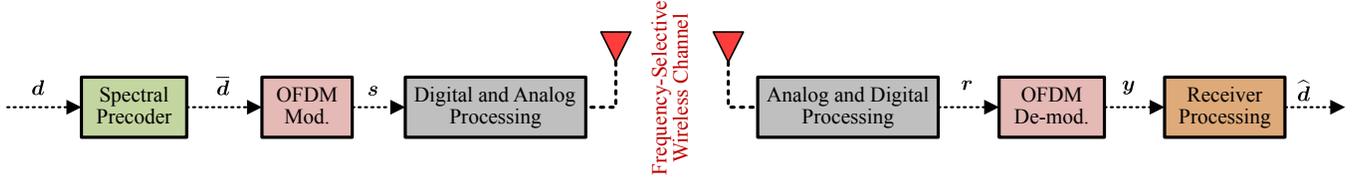}}
\end{center}
\caption{{Simplified block diagram of the OFDM transceiver with spectral precoding.}}
\vspace*{-2mm}
\label{fig1__block_diag_tx_rx}
\end{figure*}

The unwanted OOBE due to the OFDM frequency-domain signal 
$\vec{d}$ 
at the $M$ considered discrete frequency points 
$\vec{\nu}=\left[\nu_1,\ldots,\nu_{M}\right]$  
can be described by 
$\vec{p}(\vec{\nu})\, = \mat{A}\vec{d}.$ 

We now define 
$\mat{A}[m,:] \coloneqq \vec{a}\left(\nu_m\right)^\trans \in \Cm^{1 \times N }$, 
where $\mat{A}\left[m,k\right] \coloneqq a(\nu_m,k)$ 
can be derived in discrete form as \cite{DeBeek2009SculptingPrecoder}: 
\begin{align} 
\label{eqn:definition_of_a_nu_k_simplified_without_diric_function}
a(\nu_m,k) &=\left(\nicefrac{1}{\sqrt{N}} \right) \exp\left(j\pi\frac{\left(\nu_m-k\right)}{N} \left(N_{\rm CP}-N+1\right)\right) \nonumber \\ 
		  & \hspace{4mm} \cdot   \frac{\sin{\left(\pi\frac{\left(\nu_m-k\right)}{N}\left(N+N_{\rm CP}\right)\right)}}{\sin{\left(\pi\frac{\left(\nu_m-k\right)}{N}\right)}} ,
\end{align}
\sk{where $N_{\rm CP}$ corresponds to cyclic prefix length in samples.}

\vspace{-3mm}
\subsection{Spectral Precoding Methods---Prior Art}
In the sequel, we introduce the precoding methods 
proposed in \cite{DeBeek2009SculptingPrecoder} 
and  \cite{Tom2013MaskShaping}, 
which are referred to as NSP and MSP, respectively.

\subsubsection{Notch Based Spectral Precoding (NSP)}
In \cite{DeBeek2009SculptingPrecoder}, 
the constrained least-squares optimization problem for 
the memoryless notch based spectral precoder $\vec{G}_{\rm sp}$ reads
\begin{equation}
\begin{aligned}
& \underset{\overline{\vec{d}}}{\text{minimize}} 
& & \left\| \vec{d} - \overline{\vec{d}} \right\|_2^2 
& \text{subject to}
& & \vec{A}\overline{\vec{d}} = \vec{0}_{M\times1} \ ,\\
\end{aligned}
\end{equation}
such that the precoded OFDM symbol in frequency domain is given by 
$\overline{\vec{d}} = \vec{G}_{\rm sp} \vec{d}$,
where the spectral precoder is given by 
$\vec{G}_{\rm sp} = \vec{I}_N - \vec{A}^{\rm H}\left(\vec{A}\vec{A}^{\rm H}\right)^{-1}\vec{A}$.

\subsubsection{Mask Compliant Spectral Precoding (MSP)}
In \cite{Tom2013MaskShaping}, 
the work reported in \cite{DeBeek2009SculptingPrecoder} 
was extended such that the only SEM-constraint needs to be fulfilled, \ie,
\begin{equation} 
\label{eqn:original_msp_optimization_problem}
\begin{aligned}
& \underset{\overline{\vec{d}}}{\text{minimize}} 
& & \left\| \vec{d} - \overline{\vec{d}} \right\|_2^2 
& \text{subject to}
& & \left|\vec{A}\overline{\vec{d}}\right|^2 \leq \boldsymbol{\gamma} \ ,\\
\end{aligned}
\end{equation}
where the absolute value squared and inequality constraint is component-wise and $\boldsymbol{\gamma} \in \Rm^{M \times 1}$. 
The solution could not be expressed in the closed-form 
and thereby \cite{Tom2013MaskShaping} proposed to solve this MSP problem 
via a generic quadratic programming solver. 

\vspace{-2mm}
\section{Low-Complexity Algorithms for LS-MSP} \label{sec:general_spectral_precoding}

\sk{
In this section, we develop three first-order algorithms to solve the LS-MSP problem for catering to various hardware architectures in terms of support of computational complexity, memory efficiency, and latency. The proposed algorithms offer low complexity compared to second-order algorithms, \eg, interior-point based methods, for instance, see \cite{Beck2017}. 
In particular, the first two are based on the consensus ADMM and POCS, 
while the third one is derived 
by employing the coordinate descent scheme 
of a dual variable and capitalizing on the closed-form knowledge of the $\rank$ 1 constraint.

Based on our key observation, 
we rewrite the problem \eqref{eqn:original_msp_optimization_problem} as described below. 

\begin{observation} \label{lemma:msp_problem_with_rank1_constraints}
The constraint 
in \eqref{eqn:original_msp_optimization_problem} can be decomposed into 
$M$ $\rank$ 1 constraints such that the LS-MSP problem in \eqref{eqn:original_msp_optimization_problem} becomes
\begin{equation} \label{eqn:msp_problem_with_rank1_constraints}
\begin{aligned}
& \underset{\overline{\vec{d}}}{\text{minimize}} 
& & \left\| \vec{d} - \overline{\vec{d}} \right\|_2^2  \
& \text{s.t.}
& & \overline{\vec{d}}^\herm \; \overline{\mat{A}}_m \; \overline{\vec{d}} \leq \gamma_m \ \ \forall m , 
\end{aligned}
\end{equation}
where $\Cm^{N \times N} \ni \overline{\mat{A}}_m = \vec{a}\left(\nu_m\right)^* \vec{a}\left(\nu_m\right)^\trans$ and $\rank\{\overline{\mat{A}}_m \} \!=\! 1$.
\end{observation}
}

\vspace{-4mm}
\subsection{Efficient First-Order Algorithms for LS-MSP}

If $M=1$, the orthogonal projection onto the $\rank$ $1$ quadratic constraint is obtained in closed-form as described in the following \sk{theorem. The definitions of proximal operator $ {\prox}_{ \mathcal{X}_\mathcal{C}}\left( \vec{x} \right)$,  projection operator $\Pi_{\mathcal{C}} \left( \vec{x} \right)$, and the characteristic/indicator function $\mathcal{X}_\mathcal{C}\left( \vec{z} \right)$ are omitted and referred to, \eg, \cite{Parikh2013}\cite{Beck2017}, due to space constraint. 

\begin{theorem}[projection onto the $\rank$ 1 quadratic constraint] \label{theorem:proj_rank1_quadratic_constraint}
Let $\mathcal{C} \subseteq \Cm^{N \times 1}$ and $\mathcal{C} \neq \emptyset$ given by $\mathcal{C} = \left\{\vec{x} \in  \Cm^{N \times 1}: \vec{x}^\herm {\mathbfcal{A}} \vec{x} - b \leq  0 \right\}$, where $ \Cm^{N \times N} \ni {\mathbfcal{A}} = \vec{u} \vec{u}^\herm$ is $\rank$ 1 matrix and $b \in \Rm_+$, then
the proximal operator 
\begin{align} \label{eqn:projection_operator_rank1_quadratic_constraint}
    {\prox}_{\mathcal{X}_\mathcal{C}}\left( \vec{x} \right) = \Pi_{\mathcal{C}} \left( \vec{x} \right) = \vec{x} + \left( \frac{\sqrt{b} - \left| \vec{u}^\herm \vec{x} \right|}{\left\| \vec{u} \right\|_2^2 \left| \vec{u}^\herm \vec{x} \right|} \right) \vec{u} \left( \vec{u}^\herm \vec{x} \right) \ . 
\end{align}
\end{theorem}

\begin{proof}
The proof is omitted due to space constraint, but we refer to Section III-A in  \cite{Huang2016Consensus-ADMMProgramming} for the projection result. 
\end{proof}
}

If $M > 1$, then no closed form is known yet. Hence, in the subsequent sections, we propose and elucidate low-complexity algorithms to solve the LS-MSP problem.  

\subsubsection{POCS Based LS-MSP Solution (POCS algorithm)}

One of the simplest and low-complexity first-order methods to obtain 
a solution to the problem in \eqref{eqn:msp_problem_with_rank1_constraints} 
is by finding a solution to the intersection of convex sets, 
where each convex set corresponds to the ${\rm rank}=1$ quadratic inequality, \ie,
\vspace{-0.5mm}
\begin{equation} 
\label{eqn:problem_formulation_of_msp__M_convex_sets}
\begin{aligned}
& \underset{\overline{\vec{d}}}{\text{minimize}} \!
& & \left\| \vec{d} - \overline{\vec{d}} \right\|_2^2  
& \text{s.t.}
& & \overline{\vec{d}} \in \mathcal{C}_1 \bigcap  \mathcal{C}_2 \bigcap \cdots \bigcap \mathcal{C}_{M} ,
\end{aligned}
\end{equation} 
where the convex set is given by 
$\mathcal{C}_m = \left\{\overline{\vec{d}} : g_m\left(\overline{\vec{d}}\right) \leq  0 \right\}$.

The POCS based solution of \eqref{eqn:problem_formulation_of_msp__M_convex_sets} reads \cite{Combettes2011}
\begin{align}
\overline{\vec{d}}^{\left(i\right)} = \Pi_{\mathcal{C}_{M}} \left( \cdots \Pi_{\mathcal{C}_{1}} \left( \overline{\vec{d}}^{\left(i-1\right)} \right) \right),
\end{align}
where $\overline{\vec{d}}^{\left(0\right)} = \vec{d}$ and $\Pi_{\mathcal{C}_{m}}(\cdot) = {\prox}_{\mathcal{X}_\mathcal{C}}(\cdot)$ is given in \eqref{eqn:projection_operator_rank1_quadratic_constraint}. \sk{POCS converges to a common point of an intersection of the closed convex sets (assuming intersection is non-empty), \eg, see \cite{Combettes2011}}.

\subsubsection{Consensus ADMM Based LS-MSP Solution (ADMM algorithm)}

In our second proposal, 
we utilize ADMM with consensus optimization to solve the LS-MSP problem.  
\begin{equation*}
\begin{aligned}
& \underset{\overline{\vec{d}}, \vec{y}_m \in \mathbb{C}^{N \times 1}}{\text{minimize}} 
& & f\left(\overline{\vec{d}}\right) + \sum_{m=1}^M \mathcal{X}_{\mathcal{C}_m}\left(\vec{y}_m\right) 
& \text{s.t.}
& & \vec{y}_m = \overline{\vec{d}} \ \ \forall m , 
\end{aligned}
\end{equation*}
where we define $f\left(\overline{\vec{d}}\right) \coloneqq \left\| \vec{d} - \overline{\vec{d}} \right\|_2^2$ and the characteristic function  $\mathcal{X}_{\mathcal{C}_m}\left(\vec{y}_m\right)$  with the $\rank$ 1 constraint set is given by $\mathcal{C}_m = \left\{\vec{y}_m : \vec{y}_m^\herm \; \overline{\mat{A}}_m \; \vec{y}_m - \gamma_m \leq  0 \right\}$, which implies that $\mathcal{X}_{\mathcal{C}_m}\left(\vec{y}_m\right)$ is zero if $\vec{y}_m \in \mathcal{C}_m$ otherwise $\infty$ when $\vec{y}_m \notin \mathcal{C}_m$. 

The scaled-form consensus ADMM algorithm for the above problem can be expressed as \cite{Boyd2010}
\begin{align}
    \label{eqn:scaled_consensus_admmm_step1}
    \overline{\vec{d}} &\leftarrow \arg \underset{\overline{\vec{d}} }{ \min } \  f\left( \overline{\vec{d}} \right) + \rho \sum_{m=1}^M \left\| \vec{y}_m - \overline{\vec{d}} + \vec{z}_m   \right\|_2^2 \\
    \label{eqn:scaled_consensus_admmm_step2}
    \vec{y}_m &\leftarrow \arg \underset{\vec{y}_m }{ \min } \ \mathcal{X}_{\mathcal{C}_m}\left(\vec{y}_m\right) + \rho \left\| \vec{y}_m - \overline{\vec{d}} + \vec{z}_m   \right\|_2^2 \ \forall m\\
    \label{eqn:scaled_consensus_admmm_step3}
    \vec{z}_m &\leftarrow  \vec{z}_m + \vec{y}_m - \overline{\vec{d}}  \ \ \forall m=1,\ldots,M \ .
\end{align}

\begin{algorithm} 
\caption{ADMM LS-MSP algorithm for OOBE reduction}
 \begin{algorithmic}[1] \label{alg:consensus_admm_msp} 
 	\renewcommand{\algorithmicrequire}{\textbf{Inputs:}}
 	\renewcommand{\algorithmicensure}{\textbf{Output:}}
 	\REQUIRE $\vec{d} \in \mathbb{C}^{N \times 1}$, $\left\{ \gamma_m; \ \vec{a}\left(\nu_m\right) \right\}_{m=1}^M$, and $\rho > 0 \in \mathbb{R}$
 	\ENSURE  $\overline{\vec{d}}^{\left(I\right)} \in \mathbb{C}^{N \times 1}$ 
 \\ \textit{Initialization}: $\vec{y}_m^{(0)} = \vec{0}_{N \times 1}$ and $\vec{z}_m^{(0)} = \vec{0}_{N \times 1}$ 
  \FOR {$i = 1,2,\ldots,I$ } 
  \begin{subequations} \label{eqn:admm_algo_ver1}
  \STATE 
  \begin{flalign} \label{eqn:psp_admm_algo_update_delta_d_ver1}
  	\overline{\vec{d}}^{\left(i\right)} \!=\! \frac{1}{\left(1 + \rho M\right)} \left[ \vec{d} + \rho \sum_{m=1}^M \left(\vec{y}_m^{\left(i-1\right)} + \vec{z}_m^{\left(i-1\right)} \right)  \right] 
  \end{flalign} 
  \STATE 
  \begin{flalign} \label{eqn:psp_admm_algo_update_delta_ym_ver1}
  	\vec{y}_m^{\left(i\right)} 
  	   &=  {\prox}_{\mathcal{X}_{\mathcal{C}_m}}\left(\overline{\vec{d}}^{\left(i\right)} - \vec{z}_m^{\left(i-1\right)} \right) \ \  \{\textrm{cf.} \: \eqref{eqn:projection_operator_rank1_quadratic_constraint}\} 
  \end{flalign} 
  \STATE \begin{flalign} \vec{z}_m^{\left(i\right)} 
  &= \vec{z}_m^{\left(i-1\right)} + \vec{y}_m^{\left(i\right)} - \overline{\vec{d}}^{\left(i\right)}  \end{flalign} 
  \end{subequations}
  \ENDFOR
\end{algorithmic} 
\end{algorithm}

In the first step of \sk{our proposed ADMM LS-MSP algorithm}, taking the derivative with respect to $\overline{\vec{d}}$ and setting to zero, \ie, 
\begin{align*}
0 &\in \frac{\partial}{\partial \overline{\vec{d}}^*} \left\{ \left\| \vec{d} - \overline{\vec{d}} \right\|_2^2 + \rho \sum_{m=1}^M \left\| \vec{y}_m - \overline{\vec{d}} + \vec{z}_m   \right\|_2^2 \right\} \\
\Rightarrow 0 = & -\left(\vec{d} - \overline{\vec{d}} \right) - \rho \sum_{m=1}^M \left(\vec{y}_m + \vec{z}_m - \overline{\vec{d}} \right)
\end{align*}
yields \eqref{eqn:psp_admm_algo_update_delta_d_ver1}. The second step is a projection operator onto the rank 1 quadratic constraint (cf. Theorem \ref{theorem:proj_rank1_quadratic_constraint}) yielding \eqref{eqn:psp_admm_algo_update_delta_ym_ver1}. Algorithm \ref{alg:consensus_admm_msp} summarizes the proposed recipe for the ADMM based MSP, where $I$ denotes the total number of iterations.

\subsubsection{Semi-Analytic LS-MSP Solution (SSP algorithm)}
In our third proposal, we derived an optimal semi-analytical algorithm, dubbed as SSP, based on the Karush-Kuhn-Tucker (KKT) conditions \cite{Boyd2004ConvexOptimization} for the constrained optimization \eqref{eqn:msp_problem_with_rank1_constraints}. 
\sk{We form a Lagrangian of \eqref{eqn:msp_problem_with_rank1_constraints} by introducing the Lagrange multipliers  $\{\mu_m\}$ as follows
\vspace{-3mm}
$$
L\left(\overline{\vec{d}}, \left\{\mu_m\right\}\right) 
= \left\| \vec{d} - \overline{\vec{d}} \right\|_2^2 + \sum \limits_{m=1}^{M} \mu_m \left(  \overline{\vec{d}}^{\herm} \overline{\mat{A}}_m  \overline{\vec{d}} - \gamma_m \right).$$ 
\vspace{-1mm}
Utilizing the KKT conditions, the stationarity condition yields \eqref{eqn:primal_solution_ssp__d_bar} and the Lagrange multipliers $\{ \mu_m \}$ are obtained iteratively in the coordinate descent set-up \cite{Luo1992} as outlined in Algorithm \ref{alg:solution_ssp_ver1}.
We omit the related proofs due to space constraint and defer to future work.}

\begin{algorithm} 
\caption{SSP LS-MSP algorithm for OOBE reduction}
 \begin{algorithmic}[1] \label{alg:solution_ssp_ver1} 
 	\renewcommand{\algorithmicrequire}{\textbf{Inputs:}}
 	\renewcommand{\algorithmicensure}{\textbf{Output:}}
 	\REQUIRE $\vec{d} \in \Cm^{N \times 1}$, $\left\{\gamma_m; \ \lambda_1^m = \|\vec{a}\left(\nu_m\right)\|_2^2\right\}_{m=1}^{M}$.
 	\ENSURE  $\left\{\overline{\vec{d}} \in \Cm^{N \times 1} \right\}$ 
 \\ \textit{Initialization}:\\  $\left\{\mu_m^{(0)} = \frac{1}{\lambda_1^m} \left(\left| \vec{a}\left(\nu_m\right)^{\tran} \vec{d} \right| \sqrt{\left( \frac{\lambda_1^m}{\gamma_m} \right)} - 1 \right)\right\}$ $\forall m$; $\phi = 0$.
  \FOR {$i = 1,2,\ldots,I$}
  \FOR {$m = 1,\ldots,M$}
  \begin{subequations} \label{eqn:ssp_algo_ver1}
  \STATE 
  \begin{flalign} 
    \label{eqn:ssp_Gm_inverse_sum_o_rank1}
    \mat{G}_{\backslash m}^{-1}  &= \left( \mat{I}_N + \sum \limits_{n=1 \backslash m}^{M} \mu_n^{(i-1)} \overline{\mat{A}}_n  \right)^{-1} \\
    \alpha_1 &= \vec{a}\left(\nu_m\right)^{\trans} \mat{G}_{\backslash m}^{-1} \ \vec{d} \\
  \alpha_2 &= \vec{a}\left(\nu_m\right)^{\trans}  \mat{G}_{\backslash m}^{-1} \ \vec{a}\left(\nu_m\right)^{*} 
  \end{flalign} 
  \STATE 
  \begin{flalign}   
   \mu_m^{(i)} &= \Re \left\{ \frac{\alpha_1 \exp\left(-\iota \phi\right) - \sqrt{\gamma_m}}{\sqrt{\gamma_m} \ \alpha_2}  \right\}
   \savesubeqnumber
  \end{flalign}    
  \end{subequations}
  \ENDFOR
  \ENDFOR
 \RETURN 
 \vspace{-10mm}
  \addtocounter{equation}{-1}
 \begin{subequations}
 \recallsubeqnumber
\begin{flalign} \label{eqn:primal_solution_ssp__d_bar} \overline{\vec{d}} 
= \left( \mat{I}_N + \sum \limits_{m=1}^{M} \mu_m^{(I)} \  \overline{\mat{A}}_m \right)^{-1} \vec{d} \end{flalign} 
\end{subequations}
\end{algorithmic} 
\end{algorithm}

\vspace{-2mm}

\subsection{Complexity and Latency Analysis}
The dominating online algebraic complexity of {POCS} and ADMM LS-MSP algorithms are in the computation of the $\prox$ operator, which is in the order of $\BigOh\left(N\right)$, in terms of complex multiplications particularly, per given $m$-th discrete frequency and $i$-th iteration. 
The main complexity of SSP is due to the inverse of the sum of $\rank$ 1 matrices, \eg, \eqref{eqn:ssp_Gm_inverse_sum_o_rank1}. 
However, no online matrix inverse is required since one can utilize the Sherman-Morrison formula recursively \cite{Miller1981OnMatrices}, 
which is in the order of $\BigOh\left(MN^2\right)$  per $i$-th iterations. \sk{
Table \ref{table:complexity_comparison} compares the complexity of various SEM constrained precoding schemes. Moreover, depending on the considered hardware architectures of the real-time radios, the latencies to run the algorithm are also suggested for the completeness and deferred to the future work for the detailed analysis}. 

\begin{table} 
\sk{
\begin{center}
\caption{Complexity comparison of various methods per iteration}\label{table:complexity_comparison}
\vspace*{-2mm}
\scalebox{0.5}{
\resizebox{\textwidth}{!}{%
	\begin{tabular}{|l||c|c|} \hline
	\textbf{Method}   & \textbf{Complexity} &  \textbf{Latency}\\ \hline
    MSP \cite{Tom2013MaskShaping}   & $\BigOh\left(N^{4.5}\right)$ \cite{Kumar2016} & -- \\ \hline
    POCS                            & $\BigOh\left(MN\right)$ & High \\ \hline
    ADMM                            & $\BigOh\left(MN\right)$ but $M$ frequencies processing can be parallelized & Medium  \\ \hline
    SSP                            & $\BigOh\left(MN^2\right)$  & Low \\ \hline
	\end{tabular}
}
}
\end{center}
}
\end{table}

\section{Performance Evaluation} \label{sec:simulation_results}
In this section, we evaluate the performance of the proposed algorithms for MSP 
utilizing a 5G NR 
compliant link-level simulator and compare the performances 
with the conventional precoding algorithms, accordingly. 

\noindent \textit{Performance Metrics and Simulation Parameters}:\\
We analyze both in-band and out-of-band (OOB) distortion. 
In this paper, the in-band distortion is quantified in terms of EVM, BLER and (normalized) throughput, 
whereas the OOB distortion is quantified in terms of the operating band unwanted emissions (referred to as SEM) and (conducted) ACLR, see \cite{3GPPTS38.1042018NRReception} for the definitions. 
Specifically, 
\sk{we considered ACLR corresponding to the $1$st adjacent carrier, 
where the minimum requirement is $45$ dB \cite{3GPPTS38.1042018NRReception}}. 
It is worth highlighting that these ACLR requirements are for the complete radio chain, \ie, measurements need to be performed at the antenna connector. 
Thus, spectrum shaping may have some aggressive ACLR and SEM requirements to meet the minimum requirements at the antenna connector.

The key simulation parameters for the physical downlink shared channel (PDSCH) with type-A 
and the investigated test scenarios are summarized in Table \ref{tasim}, \eg, see \cite{3GPPTS38.1042018NRReception}\cite{3GPPTS38.2112018NRModulation} for the detailed NR physical layer. 
We have considered $15$ KHz subcarrier spacing for the NR numerology unless otherwise mentioned. Furthermore, no supporting signals are transmitted besides PDSCH along with the demodulation reference signal (DMRS) for the practical channel and noise variance estimation at the user equipment (UE) side. 
Notice that for simulations purpose, we have considered relatively narrow $5$ MHz channel bandwidth even though the proposed methods can be employed for very large bandwidths.

In addition to the parameters given in Table \ref{tasim}, the discrete frequencies for SEM compliant precoders are selected as $\nu \in \left\{ \mp 5010, \mp 4995, \mp 2565, \mp 2550 \right\}$ KHz, where the negative and positive frequencies correspond to the left and right side of the OOB of the occupied signal spectrum, respectively. 
Notice that these discrete frequencies can be asymmetrically selected for the OOBE suppression. The considered SEM, referred to as SEM 1 is $\vec{\gamma}_{{\rm SEM}\; 1} = \left[ -75, -75, -65, -65 \right]$ $\nicefrac{\rm{dBm}}{100 \ \rm{KHz}}$, corresponding to left/right side of the signal spectrum. Furthermore, we have considered another target SEM 2, \ie, $\vec{\gamma}_{{\rm SEM}\; 2} = \left[ -85, -85, -75, -75 \right]$ $\nicefrac{\rm{dBm}}{100 \ \rm{KHz}}$ for ACLR and EVM performance, particularly.

We found a suitable $\rho = 10$ for consensus ADMM algorithm (cf. Algorithm \ref{alg:consensus_admm_msp}) based on our numerical results, not depicted due to space constraint. For the MSP solution, we have employed CVX wrapper with SeDuMi solver \cite{Grant2014CVX:Beta}. 
\begin{table} 
\begin{center}
\caption{Simulation Parameters for FDD NR (rel-15) PDSCH Type-A}\label{tasim}
\vspace*{-2mm}
\scalebox{0.5}{
\resizebox{\textwidth}{!}{%
	\begin{tabular}{|l||c|c|c|} \hline
	\textbf{Parameters}   & \textbf{Test $1$}   & \textbf{Test $2$}  &   \textbf{Test $3$}    \\ \hline
	\multicolumn{1}{|l||}{Subcarrier Spacing}   & \multicolumn{3}{c|}{15 KHz} \\ \hline
	\multicolumn{1}{|l||}{Carrier BW (PRB alloc.)}      & \multicolumn{3}{c|}{$5$ MHz ($25$ PRBs)}             \\ \hline
	\multicolumn{1}{|l||}{Modulation}    & {16QAM} & {64QAM} & {adaptive (10\% BLER)} \\ \hline 
	\multicolumn{1}{|l||}{Code-rate}     & \multicolumn{2}{c|}{$\nicefrac{1}{2}$  $\nicefrac{5}{6}$} &  {adaptive (10\% BLER)}  \\ \hline
	Channel Model       & \multicolumn{2}{c|}{TDL-A ($30$ns, $10$Hz) }  & TDL-A ($300$ns, $100$Hz)   \\ \hline
	\multicolumn{1}{|l||}{Channel \& Noise power}     & \multicolumn{3}{c|}{Practical LMMSE based} \\ \hline
	\multicolumn{1}{|l||}{HARQ max transmissions}     & \multicolumn{3}{c|}{4 (3 max retransmissions with rv $\{0,2,3,1\}$)}  \\ \hline
    \multicolumn{1}{|l||}{Other Information}         		 & \multicolumn{3}{c|}{LDPC encoder \& decoder; no other impairments}   \\ \hline	
	\end{tabular}
	}
}
\end{center}
\end{table}

\noindent \textit{Simulation Results}:\\
In Fig. \ref{fig:aclr_vs_iter__all_methods__sel_iter}, we present the OOB performance in terms of ACLR against iterations corresponding to test 3 in Table \ref{tasim} (similar result for other tests). Evidently, SSP converges in $3$ iterations for both target masks to achieve the same ACLR results as rendered by MSP. The ADMM and POCS require nearly $100$ iterations for relaxed SEM 1, while for aggressive SEM 2, ADMM and POCS require $800$ and $3000$ iterations, respectively. We can observe that for the first $10$/$100$ iterations, POCS has faster convergence compared to ADMM, which could be due to initializations. On the other hand, we have observed that CVX achieves the MSP solution in nearly $25$ iterations, but we could not manage to obtain the result for every iteration. We have also observed similar behaviour with respect to EVM against iterations. Based on these results, we now fix the iterations for POCS, ADMM, and SSP as $3000$, $800$, and $3$, respectively, for the subsequent results. Furthermore, we have numerically observed that the precoders have a negligible impact on the peak to average power ratio. 

\begin{figure*}[!htbp]
\begin{multicols}{3}
     \centering \includegraphics[width=\linewidth]{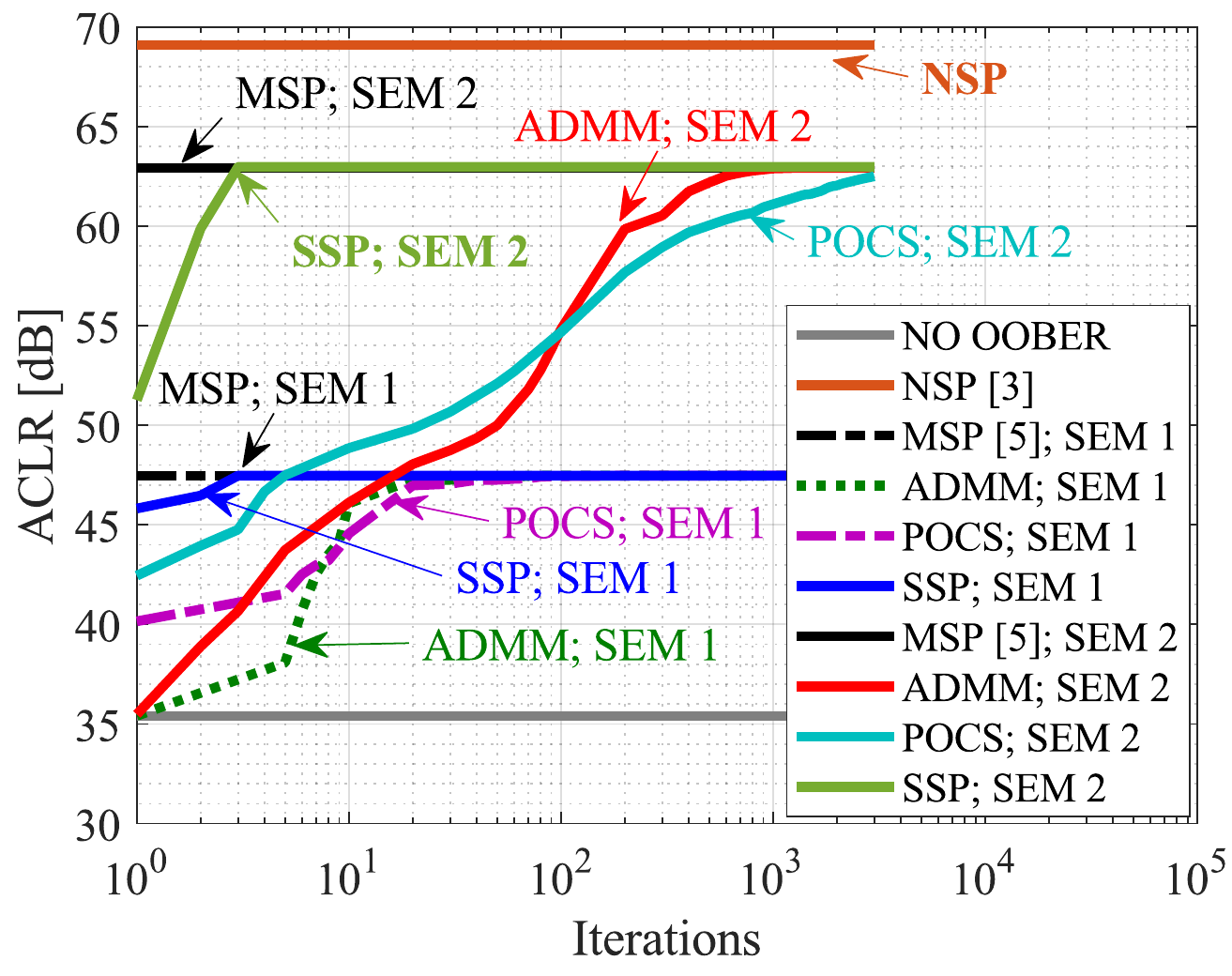}\par\caption{ACLR-vs.-Iterations along with two target masks, SEM 1 and SEM 2. \label{fig:aclr_vs_iter__all_methods__sel_iter}}
    \centering \includegraphics[width=\linewidth]{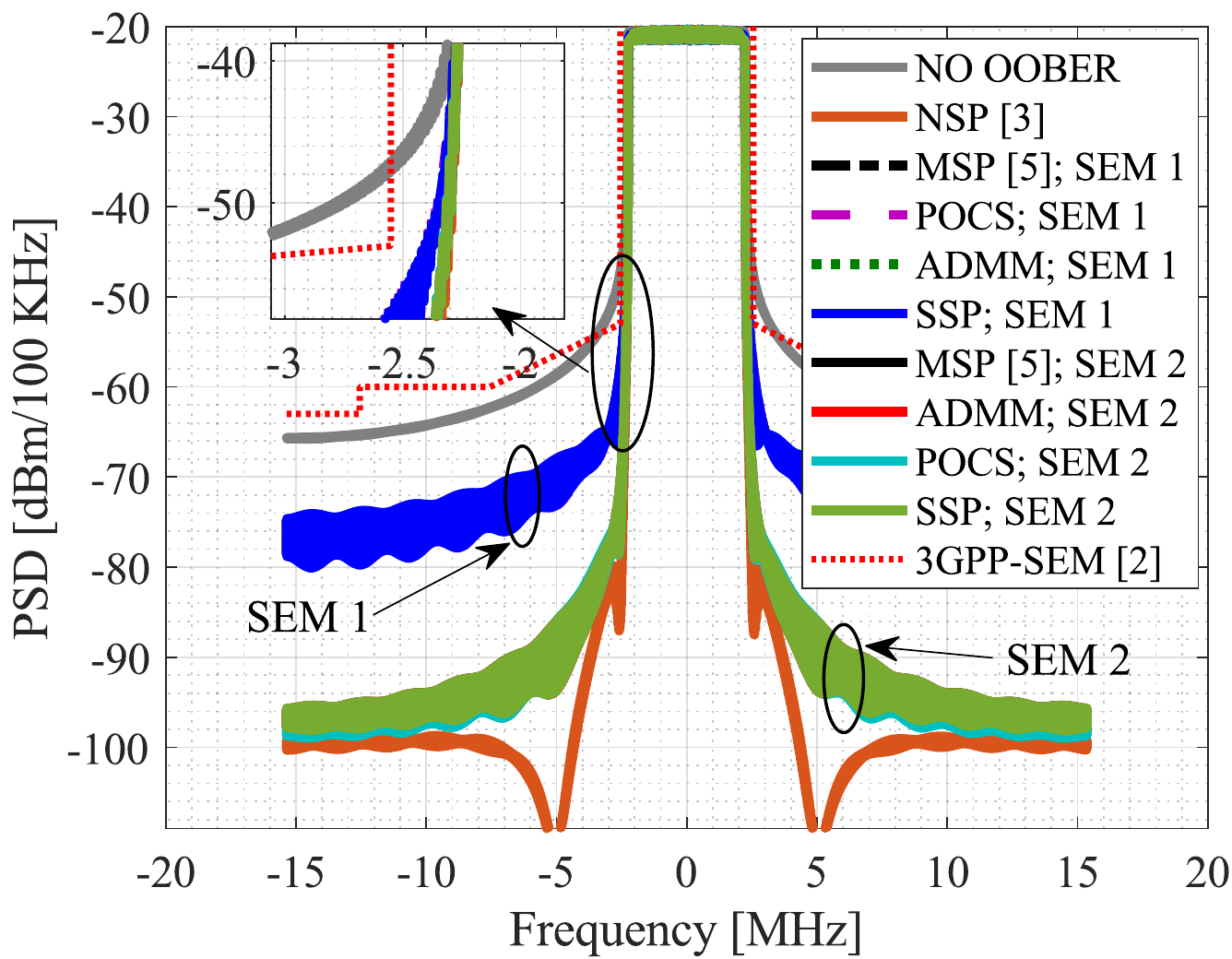}\par\caption{PSD of the various waveforms with two different target SEMs.   \label{fig:psd__all_methods__sel_iter}}
    \centering \includegraphics[width=\linewidth]{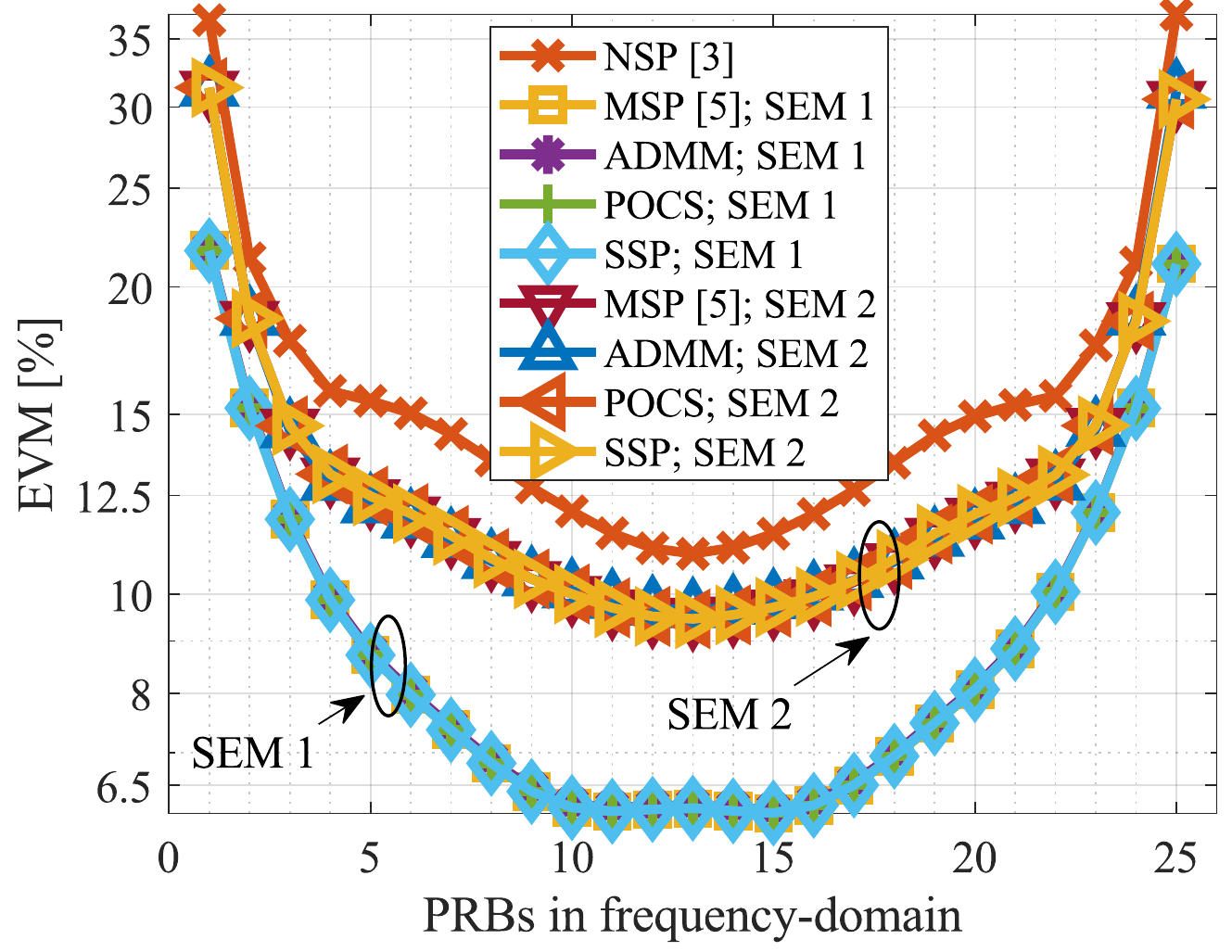}\par\caption{EVM [\%] distribution over PRBs in frequency-domain (log y-axis). \label{fig:evm_vs_prbs}}
\end{multicols}
\end{figure*}

Fig. \ref{fig:psd__all_methods__sel_iter} show the average power spectral density (PSD) versus frequency for the proposed algorithms. Moreover, the NR SEM corresponding to a medium range BS with the maximum output of $38$ dBm \cite{3GPPTS38.1042018NRReception} is also shown for the completeness, but the SEM is normalized according to the normalized transmit signal power such that the transmit signal PSD is approximately $-21.5 \; {\rm dBm/}100\; {\rm KHz}$, in the link simulations.

Fig. \ref{fig:evm_vs_prbs} present the average EVM distribution per physical resource block (PRB) \cite{3GPPTS38.2112018NRModulation} for test 3. The edge PRBs have relatively high distortion power compared to the central PRBs. It is worth to highlight that the EVM cannot be controlled in NSP unlike in the mask compliant precoders. 

\begin{figure*}[!htbp]
\begin{multicols}{3}
    \centering \includegraphics[width=\linewidth]{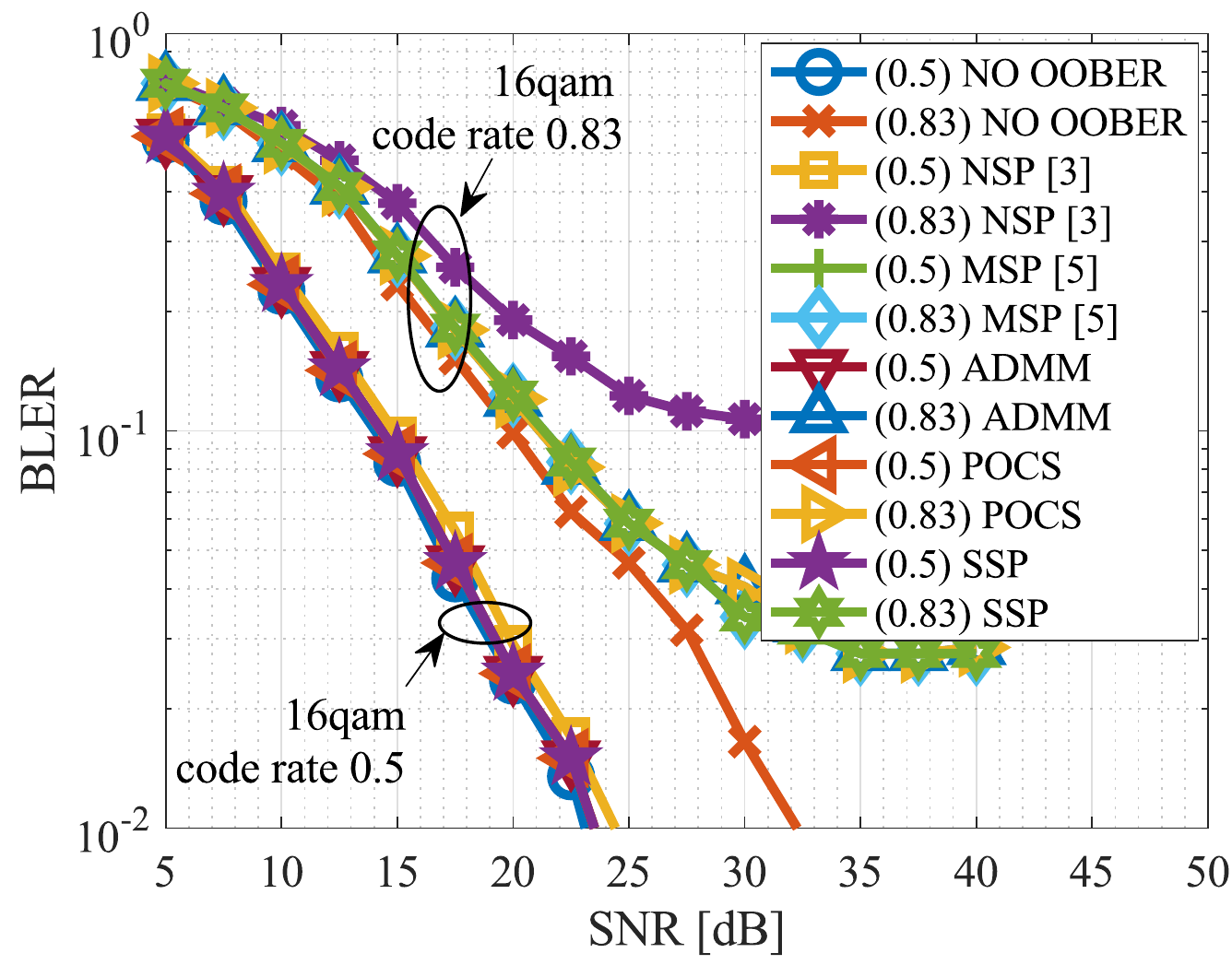}\par\caption{BLER-vs.-SNR for {16QAM}---cf. Test 1 in Table \ref{tasim}.\label{fig:bler_vs_snr__qam16__tdl_a_30ns_10hz}}
    \centering \includegraphics[width=\linewidth]{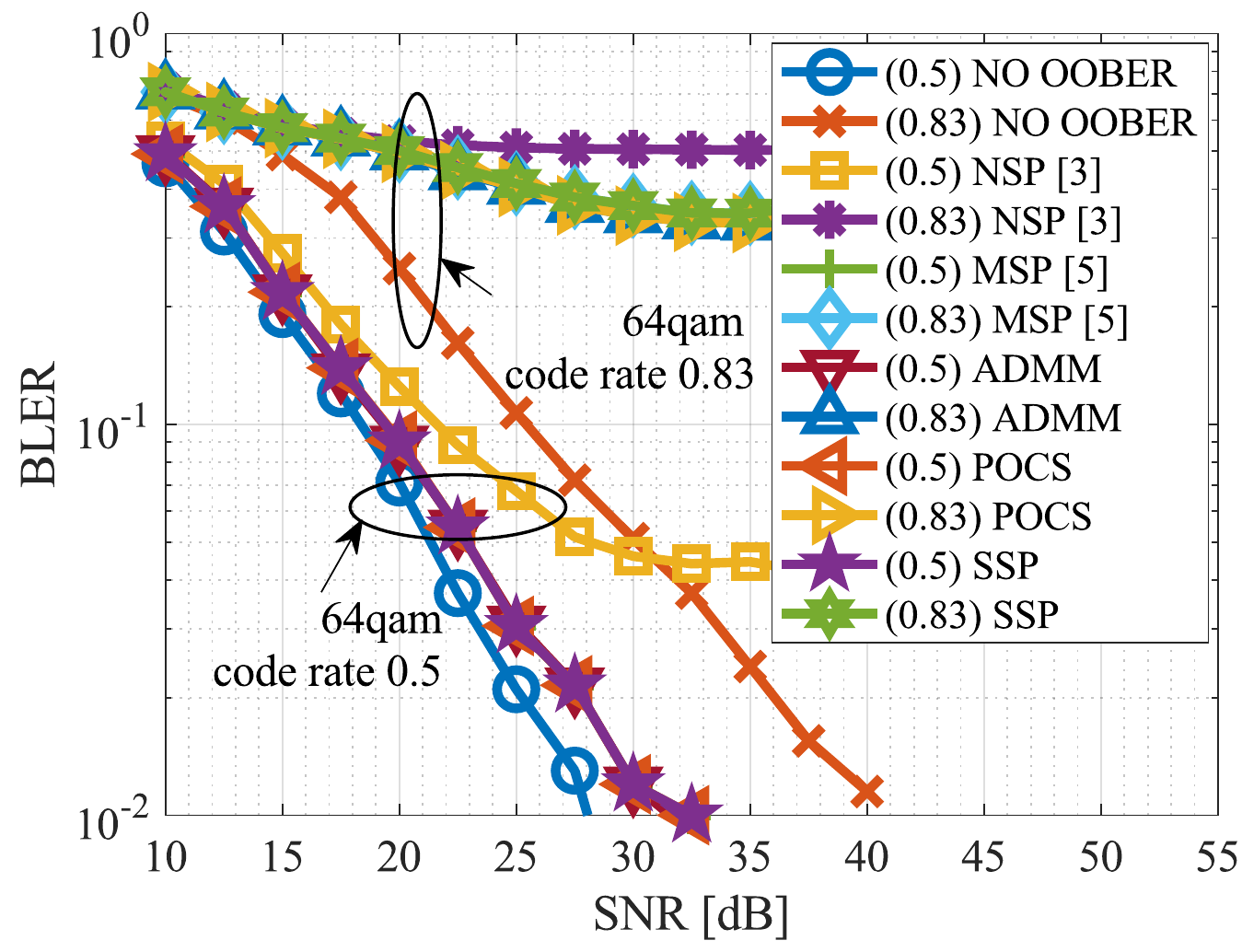}\par\caption{BLER-vs.-SNR for {64QAM}---cf. Test 2 in Table \ref{tasim}.\label{fig:bler_vs_snr__qam64__tdl_a_30ns_10hz}}
    \centering \includegraphics[width=\linewidth]{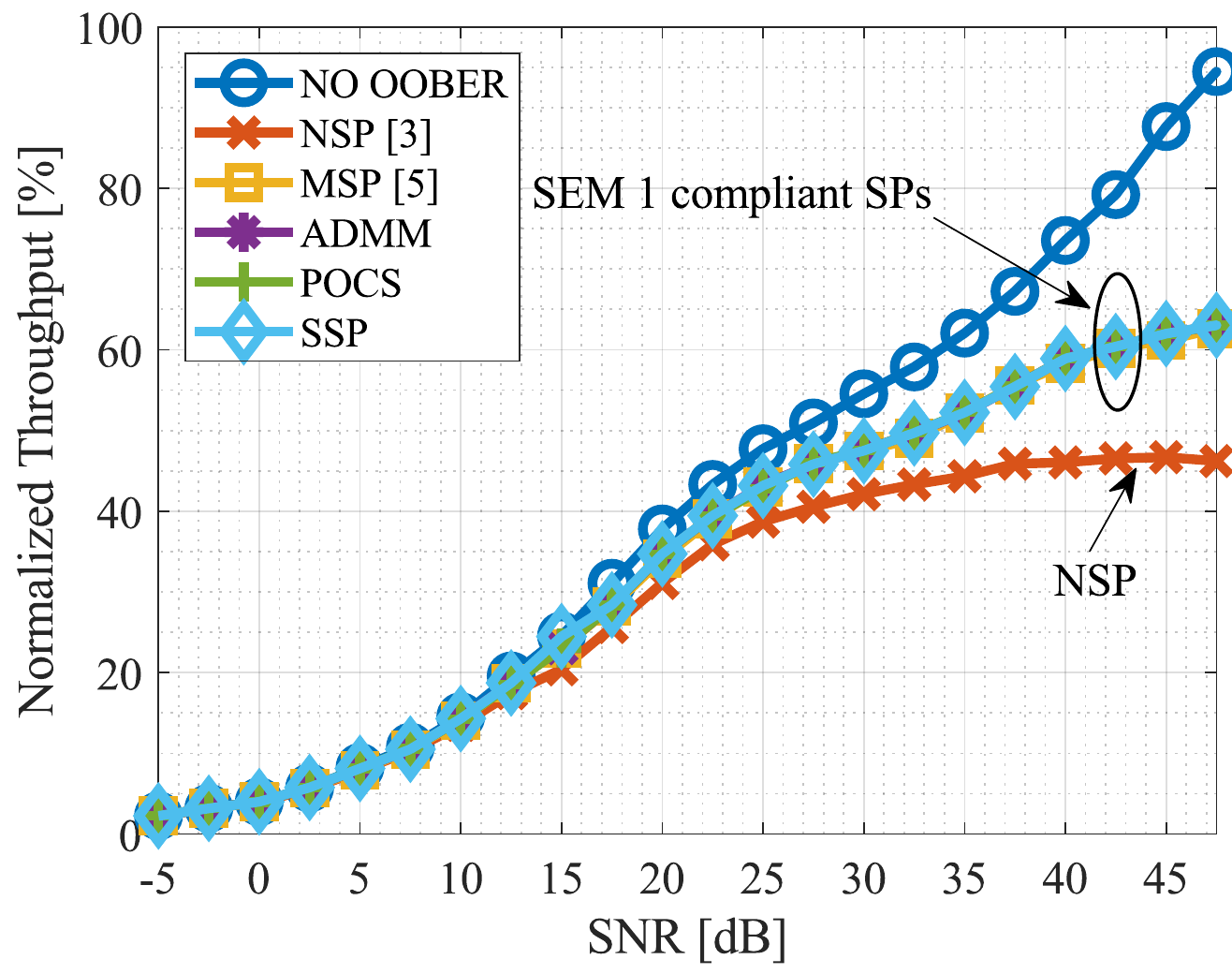}\par\caption{Normalized Throughput-vs.-SNR with link adaptation---cf. Test 3 in Table \ref{tasim}.\label{fig:tp_vs_snr__la__tdl_a_300ns_100hz}}
\end{multicols}
\vspace*{-3mm}
\end{figure*}

The (in-band) performance, in terms of BLER versus received signal-to-noise ratio (SNR) are presented in Fig. \ref{fig:bler_vs_snr__qam16__tdl_a_30ns_10hz} and Fig. \ref{fig:bler_vs_snr__qam64__tdl_a_30ns_10hz} for the fixed reference channels, \ie, fixed modulation alphabet and code rates, with 16QAM and 64QAM, respectively. For SEM compliant precoders and code rate $\nicefrac{1}{2}$ with 16QAM and 64QAM modulation alphabets, the SNR loss at 10\% BLER  is negligible and nearly $0.85$ dB, respectively, compared to no OOBE reduction scheme. In case of NSP with 64QAM $\nicefrac{1}{2}$, the SNR loss at BLER 10\% is nearly $3$ dB while error-floor has been observed around BLER 5\%. On the other hand, for SEM constrained precoders with {16QAM} $\nicefrac{5}{6}$, the SNR loss is approximately $1$ dB, but error-floor can be seen around BLER 3\%. For 64QAM $\nicefrac{5}{6}$, all the precoders render poor performance due to high transmit EVM. Fig. \ref{fig:tp_vs_snr__la__tdl_a_300ns_100hz} present the Throughput vs. SNR with link adaptation having a target of 10\% BLER. Although SEM-constrained precoders show better performance compared to NSP, there can be a potential to improve the performance of spectral precoding, notably in high SNR regime, which will be addressed in our future work.

\section{Conclusion} \label{sec:conclusion_future_work}
We presented three hardware-friendly algorithms for spectral emission mask constrained precoder
for the large-scale OFDM based systems,
namely ADMM, POCS, and 
SSP utilizing KKT conditions and a coordinate descent scheme. 
Numerical results corroborate that the proposed algorithms achieve similar performance as existing high complexity MSP \cite{Tom2013MaskShaping} utilizing generic convex optimization solvers. Furthermore, our simulation results show that SSP converges in 
typically 3 iterations with some increased computational cost yet offers lower computational cost 
than the conventional MSP.

\bibliographystyle{ieeetr}
\bibliography{ref}

\end{document}